\documentclass[11pt]{article}

\usepackage[margin=0.9in]{geometry}
\usepackage{amsmath,amssymb,amsthm,mathtools}
\usepackage{booktabs}
\usepackage{xcolor}
\usepackage{enumitem}
\usepackage{algorithm}
\usepackage{algpseudocode}
\algrenewcommand{\algorithmicrequire}{\textbf{Input:}}
\algrenewcommand{\algorithmicensure}{\textbf{Output:}}
\usepackage{microtype}
\usepackage{needspace}
\usepackage[hidelinks]{hyperref}
\usepackage{url}

\allowdisplaybreaks

\newtheorem{theorem}{Theorem}[section]
\newtheorem{lemma}{Lemma}[section]

\newtheorem{definition}{Definition}[section]

\theoremstyle{remark}

\DeclareMathOperator{\OPT}{OPT}
\DeclareMathOperator{\rank}{rank}
\newcommand{\cl}{\operatorname{cl}}
\newcommand{\revision}[1]{#1}

\title{A Refined Analysis for Matroid Secretary with Submodular Objectives}
\author{
Dennis Joyce
}
\date{}

\begin{document}
\maketitle
\begingroup
\renewcommand{\thefootnote}{}
\begin{NoHyper}
\footnotetext{East Texas A\&M University. djoyce2@leomail.etamu.edu}
\end{NoHyper}
\endgroup

\begin{abstract}
We study the matroid secretary problem with a nonnegative monotone submodular objective. Elements arrive in uniformly random order, and every acceptance decision is immediate and irrevocable. We give an $8.699$-competitive algorithm for arbitrary matroids.
 The algorithm first rejects a randomly sized initial segment of the arrival sequence as a learning set and runs submodular greedy on all elements in this set. Each later element receives a fixed weight equal to the marginal value it would have when inserted into the stored greedy sequence.
  On these fixed weights, the algorithm runs a linear-objective matroid secretary subroutine.
The algorithm uses $O(nr)$ value queries and $O(n^2)$ independence queries, where $r$ is the matroid rank.
\end{abstract}

\section{Introduction}

The matroid secretary problem models online selection under an independence constraint. 
The elements of the ground set are revealed one by one in a uniformly random order. 
When an element arrives, the algorithm may make value and independence queries involving the elements revealed so far and must immediately and irrevocably decide whether to accept it. The accepted elements must remain independent in the matroid. In the classical linear version, every element has a fixed weight and the objective is the sum of the accepted weights. 
Babaioff, Immorlica, and Kleinberg~\cite{babaioff2007} introduced this problem and gave an $O(\log r)$-competitive algorithm for arbitrary matroids, where $r$ is the matroid rank. 
Constant-competitive algorithms were obtained only recently~\cite{abdi2026,clw2026,huang2026,singla2026}. 

The submodular version replaces the linear objective by a set function with diminishing returns.
A set function $f$  is \emph{submodular} if the marginal value of an element can only decrease as the current set grows. The resulting online problem is the submodular matroid secretary problem. 
Feldman and Zenklusen~\cite{feldmanzenklusen2018} gave a black-box reduction that transforms an $\alpha$-competitive algorithm for the linear problem into an $8\alpha(\alpha+1)$-competitive algorithm for the monotone submodular problem. In particular, when the linear algorithm is $\alpha$-probability-competitive (i.e., every element in a fixed optimal offline solution has probability at least $1/\alpha$ of being accepted), their guarantee improves to $16\alpha$. 
Combining the latter reduction with the current $e$-probability-competitive linear algorithm for arbitrary matroids by Huang~\cite{huang2026} gives a $16e\approx 43.493$-competitive algorithm for the submodular version.

We obtain an $8.699$-competitive algorithm for the submodular matroid secretary problem by instantiating the reduction in~\cite{feldmanzenklusen2018} and giving a refined joint analysis. 
\begin{theorem}[Main theorem]\label{thm:main-result}
There exists a polynomial-time randomized online algorithm for the monotone submodular matroid secretary problem on arbitrary matroids that is $C$-competitive, where
\[
C=\frac{1559}{299(1-e^{-47/23})}+\frac{474}{175}
=8.6988177308\ldots.
\]
On a ground set of $n$ elements with matroid rank $r$, the algorithm uses $O(nr)$ value queries and $O(n^2)$ independence queries.
\end{theorem}

The algorithm uses two existing ingredients and a new joint analysis.
\revision{First, we use the online framework of the monotone simulated-greedy reduction of Feldman and Zenklusen~\cite{feldmanzenklusen2018}. We reject a random learning sample and run submodular greedy on it. The learning sample fixes a greedy sequence against which each later element is assigned a \emph{frozen weight}: the marginal value at which that element would enter the stored greedy sequence if it were added to the learning sample. Conditional on the learning sample, these weights are fixed before the remaining random order is exposed. The broader idea of defining weights from an initial sample appeared earlier in the work of Ma, Tang, and Wang~\cite{ma2013}.}

Second, we instantiate the reduction \cite{feldmanzenklusen2018} with the Dynamic Thinning algorithm. Dynamic Thinning, due to Chan, Lin, and Wang \cite{clw2026}, is a recent $3.1462$-probability-competitive algorithm for the linear-objective matroid secretary problem.
 Informally, its output has the same distribution as the following simple procedure: sample each element independently, run fixed-weight matroid greedy on the sampled elements, and then retain each greedy-selected element independently with a fixed probability. A precise description is given in Section~\ref{sec:DT}.
We run Dynamic Thinning on the post-learning elements with their frozen weights, and return the accepted elements.
 
\revision{The main contribution is the joint analysis of this instantiation. Note that the direction application of the reduction with Dynamic Thinning gives only a $16\cdot 3.1462\approx 50.339$-competitive guarantee for the submodular problem.
To prove the 8.699-competitive ratio, we first give an equivalent probabilistic description of the online output and conduct all subsequent analysis within this description. Whereas Feldman and Zenklusen use black-box guarantees for the linear-objective algorithm,
 we use its exact output
distribution to analyze jointly the value captured during the learning phase and the value contributed
by the final online output. This avoids some of the loss inherent in a black-box composition and yields
the improved guarantee.}
 The guarantee holds for every finite ground set and every matroid rank, and the algorithm is polynomial-time and does not require exact offline submodular optimization.

\subsection{Related work}\label{subsec:related}

The matroid secretary problem has been studied since the work of Babaioff et al.~\cite{babaioff2007,babaiokempe2018}. Earlier general guarantees depended on the matroid rank~\cite{chakraborty2012,lachish2014,feldman2018}. Recent work established constant guarantees for arbitrary matroids. Singla~\cite{singla2026} gave a $4$-competitive algorithm. In concurrent work, Abdi et al.~\cite{abdi2026} and B\'erczi et al.~\cite{berczi2026} proved the strong matroid secretary conjecture for linear matroids, both giving $e$-competitive ordinal algorithms. Huang~\cite{huang2026} subsequently obtained the exact factor $e$ for arbitrary matroids with expected polynomial query complexity.

\revision{Under a cardinality constraint, the submodular $k$-secretary problem asks the algorithm to select at most $k$ elements. This case was independently studied by Bateni, Hajiaghayi, and Zadimoghaddam~\cite{bateni2013} and Gupta et al.~\cite{gupta2010}. For monotone objectives, \cite{bateni2013} gave a $7/(1-1/e)\approx11.073$-competitive algorithm, and Feldman, Naor, and Schwartz~\cite{feldmannaorschwartz2011} improved the ratio to $(e^2+e)/(e-1)\approx5.882$. Given an offline algorithm that returns at least an $\alpha$ fraction of the cardinality-constrained optimum, where $0<\alpha\le1$, Kesselheim and T\"onnis~\cite{kesselheim2017} obtained a competitive ratio of at most $3.226/\alpha$ for every $k\ge2$, approaching $e/\alpha$ as $k$ grows. Their analysis specialized to greedy gives the stronger asymptotic ratio $3.632$ as $k\to\infty$. For nonnegative objectives that need not be monotone, Gupta et al. \cite{gupta2010} gave a constant competitive ratio, and Fujii~\cite{fujii2019} obtained the ratio $e^2(1+e)/(e-1)^2\approx9.305$.}

For the submodular matroid secretary under special matroid constraints, Gupta et al.~\cite{gupta2010} gave constant competitive ratios for non-monotone objectives under partition matroids. For monotone objectives, Ma, Tang, and Wang~\cite{ma2013} gave a $211$-competitive algorithm for laminar matroids and a $95$-competitive algorithm for transversal matroids. Feldman and Zenklusen~\cite{feldmanzenklusen2018} improved the laminar ratio to $144$. Kesselheim and T\"onnis~\cite{kesselheim2017} further improved submodular secretary matching. 

\medskip The rest of the paper is organized as follows. Section~\ref{sec:prelim} defines the model. Section~\ref{sec:algorithm} gives the online algorithm. Section~\ref{sec:process} gives an equivalent random process used only in the analysis. Section~\ref{sec:analysis} proves the competitive and query guarantees. The appendix gives improved competitive guarantees for laminar matroids and graphic matroids. 

\section{Preliminaries}\label{sec:prelim}

\subsection{Matroids and monotone submodularity}

Let $(E,\mathcal I)$ be a matroid on a ground set $E$ of $n$ elements. A set in $\mathcal I$ is called \emph{independent}. For $S\subseteq E$, let $\rank(S)$ be the maximum size of an independent subset of $S$, and let $r=\rank(E)$ be the matroid rank. For a set $S$ and an element $e$, write $S+e=S\cup\{e\}$. The \emph{closure} of $S$ is
\[
\cl(S)=\{e\in E:\rank(S+e)=\rank(S)\}.
\]
When $e\in\cl(S)$, we also say that $S$ \emph{spans} $e$. We use an independence oracle: given a set of elements, the oracle reports whether that set belongs to $\mathcal I$.

A function $f:2^E\to\mathbb R_{\ge0}$ is \emph{monotone} if
\(
f(S)\le f(T)
\) for any $S\subseteq T\subseteq E$,
and is \emph{submodular} if
\[
f(e\mid S)\ge f(e\mid T)
\qquad(S\subseteq T,\ e\notin T),
\]
where
\(
f(e\mid S)=f(S+e)-f(S)
\)
is the \emph{marginal value} of $e$ with respect to $S$. We assume $f(\varnothing)=0$. 

In the secretary setting, the elements arrive in a uniformly random order.
 A \emph{value query} on an observed set $S$ returns $f(S)$. When an element arrives, 
 the algorithm may make value queries and independence queries using only elements that have already arrived. It must then accept or reject the element immediately.
  Rejected elements cannot later be accepted. The accepted set must remain independent in $(E,\mathcal I)$.

For a monotone submodular function $f$, let
\(
O\in\arg\max\{f(S):S\in\mathcal I\}
\)
be a fixed optimal independent set. An online algorithm is $C$-\emph{competitive} if
\[
\mathbb E[f(A)]\ge \frac1C f(O),
\]
where $A$ is its random accepted set.

We fix a label order, independent of the random arrivals, to break all ties consistently. For a fixed nonnegative weight vector $\mathbf w=(w_e)_{e\in E}$, write $\OPT_{\mathbf w}(S)$ for the greedy set obtained by scanning the positive-weight elements of $S$ in decreasing weight-and-label order and adding each element that preserves independence. That is, $\OPT_{\mathbf w}(S)$ is a maximum-weight independent subset of $S$. 

\subsection{Submodular greedy}

For a set $S\subseteq E$, \emph{submodular greedy} starts with the empty set and repeatedly adds a feasible element of maximum positive current marginal value in terms of $f$. Ties are broken by the fixed label order. It stops when no feasible element has positive marginal value.

If submodular greedy chooses the ordered sequence
\[
m_1,m_2,\ldots,m_k,
\]
we write
\[
M_i=\{m_1,\ldots,m_i\},\qquad M_0=\varnothing,
\]
and call the sets $M_0,M_1,\ldots,M_k$ the \emph{greedy prefixes}. 

The distinction between the two greedy rules is important. Submodular greedy chooses an element according to its current marginal value, which changes as the chosen set grows. The notation $\OPT_{\mathbf w}(S)$ refers instead to ordinary matroid greedy under a fixed weight vector $\mathbf w$, which is a linear objective.

\subsection{Dynamic Thinning}\label{sec:DT}

We use Dynamic Thinning~\cite{clw2026} as a linear-objective matroid secretary subroutine. We only need its online implementability and its exact terminal distribution, which we restate here.

Fix a weighted matroid secretary instance with $N$ elements and a parameter $s\in(0,1)$. \revision{Draw $N$ independent uniform random variables in $(0,1)$,  and assign them in increasing order to successive arrivals as the auxiliary times.} Elements with auxiliary time at most $s$ form an initial rejected sample and initialize a \emph{reference set} $X$. The reference set is a subset of arrived elements and is different from the actual accepted set $A$.

At a later auxiliary time $t>s$, let $B=\OPT_{\mathbf w}(X)$ and $B^+=\OPT_{\mathbf w}(X+e)$ for the arriving element $e$. The algorithm acts as follows. If $e\notin B^+$, it rejects $e$. If $B^+=B+e$, it accepts $e$ with probability
\[
q(t)=\exp\left(-\frac{t-s}{s}\right).
\]
Otherwise there is a unique $b\in B$ such that $B^+=B-b+e$. In this case the algorithm accepts $e$ exactly when $b\notin A$. Every accepted arrival is inserted into both $X$ and $A$. Between two consecutive arrivals, every element of $B\setminus A$ has an independent exponential deletion clock of rate $1/s$; when its clock rings, that element is removed from $X$. These internal deletions change only the reference set and never revoke an acceptance. \revision{After each deletion, recompute $B=\OPT_{\mathbf w}(X)$ and restart independent rate-$1/s$ clocks for the elements of the current $B\setminus A$. Recompute $B$ and restart the clocks after each arrival as well. After the last arrival, continue internal deletions until time $1$.} Return $A$. 

The following result is the property we use in the submodular analysis.

\begin{lemma}[Output distribution of Dynamic Thinning~\cite{clw2026}]\label{lem:dt-law}
Fix a weighted matroid secretary instance and $s\in(0,1)$. Let $X$ be a random subset that contains every element independently with probability $s$.
Set
\begin{equation}\label{eq:gamma}
\gamma=1-\exp\left(-\frac{1-s}{s}\right).
\end{equation}
Conditional on $X$, retain every element of $\OPT_{\mathbf w}(X)$ independently with probability $\gamma$. The resulting random set has the same distribution as the online output $A$ of Dynamic Thinning with parameter $s$.
\end{lemma}

An \emph{independent $\gamma$-thinning} of a set means the random subset obtained by retaining each member independently with probability $\gamma$. Lemma~\ref{lem:dt-law} says that Dynamic Thinning outputs an independent $\gamma$-thinning of $\OPT_{\mathbf w}(X)$, where $X$ is the random reference set.
We also use its query bound as follows.

\begin{lemma}[Query bound of Dynamic Thinning~\cite{clw2026}]\label{lem:dt-query}
On an instance with $N$ elements, Dynamic Thinning uses $O(N^2)$ independence queries in the worst case. Every queried set contains only elements that have already arrived.
\end{lemma}

\section{The Online Algorithm}\label{sec:algorithm}

\revision{Algorithm~\ref{alg:main} instantiates the reduction of Feldman and Zenklusen~\cite{feldmanzenklusen2018} with Dynamic Thinning as the linear secretary subroutine.}
The algorithm has a learning phase followed by a fixed-weight (linear) secretary phase. The learning phase is used only to define the weights. All learning-phase elements are rejected.

Fix two parameters $p,s\in(0,1)$. Draw
\[
K\sim\operatorname{Bin}(n,p)
\]
independently of the arrival order. Reject the first $K$ arrivals and denote their set by $L$. Equivalently, $L$ is an independent Bernoulli-$p$ sample of $E$.

Run submodular greedy on $L$. Let $m_1,\ldots,m_k$ be the elements it selects in order, let
\[
M=\{m_1,\ldots,m_k\},
\]
and let $M_i=\{m_1,\ldots,m_i\}$ be the greedy prefixes. Store this ordered sequence and its prefixes.

Consider an element $e\notin L$. Its \emph{frozen weight} $w_e$ is defined by running submodular greedy on $M+e$. If this greedy run never selects $e$ with positive marginal value, set $w_e=0$. Otherwise, let $M_e$ be the greedy prefix present immediately before $e$ is selected and set
\begin{equation}\label{eq:frozen-weight}
w_e=f(e\mid M_e).
\end{equation}
The weight is called frozen because, once $L$ is fixed, it is determined by $L$ and $e$ alone. It is independent of the order of the later arrivals and of all later online decisions.

After the learning phase, run one Dynamic Thinning instance of parameter $s$ on the remaining $n-K$ arrivals using the fixed weights $w_e$. Elements with weight zero stay in the stream with weight zero and are never selected by the positive-weight linear greedy rule. The actual accepted set of Dynamic Thinning is the output of our algorithm.

\begin{algorithm}[t]
\caption{Monotone Submodular Dynamic Thinning}\label{alg:main}
\begin{algorithmic}[1]
\Require Random-order stream of $n$ elements, independence oracle, value oracle, parameters $p,s\in(0,1)$.
\Ensure An independent accepted set $A$.
\State Draw $K\sim\operatorname{Bin}(n,p)$.
\State Reject the first $K$ arrivals and call their set $L$.
\State Run submodular greedy on $L$, obtaining the ordered set $M$ and all greedy prefixes.
\State Initialize one Dynamic Thinning instance of parameter $s$ for the $n-K$ post-learning arrivals.
\For{each later arrival $e$}
    \State Compute its frozen weight $w_e$ from greedy on $M+e$ as in~\eqref{eq:frozen-weight}.
    \State Feed $e$, with weight $w_e$, to this Dynamic Thinning instance.
    \State Make the same accept/reject decision as this Dynamic Thinning instance.
\EndFor
\State \Return the elements accepted by Dynamic Thinning.
\end{algorithmic}
\end{algorithm}

Algorithm~\ref{alg:main} is online. The phrase ``run greedy on $M+e$'' uses only the stored learning sample and the currently arriving element. The inner Dynamic Thinning instance uses fresh internal randomness and treats the post-learning stream as its entire input.

\begin{lemma}\label{lem:fixed-weights}
Conditional on the learning sample $L$, the weights $(w_e)_{e\in E\setminus L}$ are fixed before the random order of $E\setminus L$ is exposed. The relative order of the post-learning elements remains uniformly random.
\end{lemma}

\begin{proof}
The ordered greedy set $M$ and all its prefixes are determined by $L$. For every $e\notin L$, the rule in~\eqref{eq:frozen-weight} depends only on these prefixes, the identity of $e$, the matroid restriction to $M+e$, and value queries involving $M+e$. Hence $w_e$ is a deterministic function of $L$ and $e$ once the instance is fixed. It does not depend on the positions of the other post-learning elements or on any later random choice of the algorithm.

Conditional on $L$ being the first $K$ arrivals, the remaining elements appear in a uniformly random relative order. Thus, conditional on $L$, the second phase is a valid fixed-weight matroid secretary instance.
\end{proof}

\section{An Equivalent Three-Color Process}\label{sec:process}

\revision{We analyze the online algorithm through an offline process with the same output distribution. 
Following the simulated-greedy approach \cite{ma2013,feldman2018}, we select elements according to their marginal values with respect to the learning set and then reveal their random classifications. Here we use three classes and an additional independent retention step to describe the output distribution of the algorithm.
Lemma~\ref{lem:simulation} establishes the equivalence of output distributions, and all subsequent analysis uses this process.}

Define
\begin{equation}\label{eq:beta-gamma}
\beta=(1-p)s,
\qquad
\gamma=1-\exp\left(-\frac{1-s}{s}\right).
\end{equation}

\par\medskip
\noindent\begingroup
\setlength{\fboxsep}{8pt}
\fbox{\begin{minipage}{\dimexpr\linewidth-2\fboxsep-2\fboxrule-3pt\relax}
\textbf{Three-color process.}\par\smallskip
Independently assign every element one of three colors:
\[
\Pr[\text{learning}]=p,
\qquad
\Pr[\text{reference}]=\beta,
\qquad
\Pr[\text{ignored}]=1-p-\beta.
\]
Independently give every element a Bernoulli retention bit equal to one with probability $\gamma$.

Start with
\[
M=B=A=\varnothing.
\]
Repeatedly choose, among all unprocessed elements that are feasible to add to $M$, an element $e$ of maximum positive marginal value
\[
w_e=f(e\mid M),
\]
breaking ties by the fixed label order. Only after choosing $e$ do we reveal its color.

If $e$ is learning-colored, add it to $M$. If $e$ is reference-colored, add it to $B$ when $B+e$ is independent. If it is added to $B$ and its retention bit is one, also add it to $A$. If $e$ is ignored-colored, do nothing. An element that becomes infeasible for $M$ can be discarded permanently because $M$ only grows. The process stops when no feasible element has positive marginal value.
\end{minipage}}
\endgroup
\par\medskip

The set $B$ is an auxiliary reference greedy set in the analysis. The set $A$ is its retained part and will have the same distribution as the actual online output.

\begin{lemma}[Distributional equivalence]\label{lem:simulation}
The final set $A$ in the three-color process has the same distribution as the output of Algorithm~\ref{alg:main}.
\end{lemma}

\begin{proof}
First condition on the learning-colored set $L$. The learning-colored additions in the process are exactly submodular greedy on $L$, so they produce the same ordered set $M$ as Algorithm~\ref{alg:main}.

Fix a non-learning element $e$. Before the process handles $e$, only learning-colored elements have changed $M$. Removing the other non-learning elements from the conceptual scan therefore does not change the greedy prefix present when $e$ is considered. The weight assigned to $e$ by the process is therefore exactly its frozen weight $w_e$ from~\eqref{eq:frozen-weight}.

The positive weights encountered during the process are nonincreasing. Indeed, after a learning-colored element is added to $M$, submodularity can only decrease the remaining marginal values, and matroid feasibility can only remove options. Hence, among reference-colored elements, the rule ``add $e$ when $B+e$ is independent'' is exactly fixed-weight matroid greedy under the frozen weights. Thus $B=\OPT_{\mathbf w}(X)$, where $X$ is the set of reference-colored elements.

Conditional on $L$, every non-learning element is reference-colored with probability
\(
\frac{\beta}{1-p}=s,
\)
independently of the others. Therefore $X$ is an independent Bernoulli-$s$ sample of the post-learning elements. Finally, the retention bits make $A$ an independent $\gamma$-thinning of $B$. Lemma~\ref{lem:dt-law} now identifies this distribution with that of the actual Dynamic Thinning instance in Algorithm~\ref{alg:main}.
\end{proof}

The analysis below always considers the state immediately before the color and retention bit of the current element are revealed. At that moment, its weight $w_e$, the current sets $M,B,A$, and every span or feasibility test involving these sets are already determined. This order of conditioning is important in the probability calculations.

\section{Competitive Analysis}\label{sec:analysis}

The analysis proceeds in three steps. Section~\ref{subsec:rank-charge} provides a useful inequality. Section~\ref{subsec:residual-optimum} uses this inequality to upper-bound the optimal value. Section~\ref{subsec:weight-redundancy} proves the  competitive ratio of Algorithm~\ref{alg:main}.

Immediately before the color of the current element $e$ is revealed, call the step \emph{external} if
\[
e\notin\cl(M\cup B),
\]
and \emph{internal} otherwise. 

\subsection{A rank charge}\label{subsec:rank-charge}

The nonnegative integer
\[
\rank(M\cup B)-|M|
\]
changes in a particularly simple way. At an external step, $B+e\in\mathcal I$, so a reference-colored element is added to $B$ and increases the quantity by one. At an internal step, a learning-colored element is feasible for $M$ but already lies in $\cl(M\cup B)$; adding it to $M$ decreases the quantity by one. Every other color/outcome leaves the quantity unchanged.

Because the quantity starts at zero and never becomes negative, every internal learning addition
 can be paired with a distinct earlier external reference addition. The processing weights are 
 nonincreasing, so the earlier reference-colored element has weight at least that of its paired learning-colored element. Taking expectations yields the following weighted charge.

\begin{lemma}[Rank charge]\label{lem:rank-charge}
Let the sums range over the processed steps of the three-color process. Then
\begin{equation}\label{eq:rank-charge}
p\,\mathbb E\!\left[\sum_{\mathrm{internal}} w_e\right]
\le
\beta\,\mathbb E\!\left[\sum_{\mathrm{external}} w_e\right].
\end{equation}
\end{lemma}

\begin{proof}
On every realization, pair each $-1$ change of $\rank(M\cup B)-|M|$ with a distinct earlier $+1$ change. The weight at the earlier step is at least the weight at the later step, because the processing weights are nonincreasing. Thus, pathwise,
\[
\sum_{\substack{\mathrm{internal}\text{ learning steps}}}w_e
\le
\sum_{\substack{\mathrm{external}\text{ reference steps}}}w_e.
\]
The weight and the internal/external classification are determined before the current color is revealed. An internal step is learning-colored with probability $p$, and an external step is reference-colored with probability $\beta$. Taking expectations gives~\eqref{eq:rank-charge}.
\end{proof}

\subsection{Upper bound on the optimum}\label{subsec:residual-optimum}

Fix an offline optimal independent set $O$. During the three-color process, maintain a set $Q\subseteq O$ of unprocessed optimal elements satisfying
\begin{equation}\label{eq:Q-rank}
\rank(M\cup B\cup Q)=\rank(M\cup B)+|Q|.
\end{equation}
Initially $Q=O$. Condition~\eqref{eq:Q-rank} means that adding all elements of $Q$ to $M\cup B$ increases the rank by exactly $|Q|$. In particular, every element of $Q$ is external and feasible to add to $M$.

The update is as follows. If the current element $e$ belongs to $Q$, remove it from $Q$ after it is processed, regardless of its color. If $e\notin Q$ and is not added to either $M$ or $B$, leave $Q$ unchanged. Now suppose that $e\notin Q$ and is added to $M$ or to $B$. At an internal step, this addition does not change the closure of $M\cup B$, so~\eqref{eq:Q-rank} continues to hold and no element is removed from $Q$. At an external step, the addition increases $\rank(M\cup B)$ by one. If~\eqref{eq:Q-rank} still holds after the addition, leave $Q$ unchanged. Otherwise, choose an element $o\in Q$ whose removal restores~\eqref{eq:Q-rank}, and replace $Q$ by $Q-o$. Such an element exists by the matroid exchange property, and removing one element is sufficient because $\rank(M\cup B)$ has increased by only one.

\begin{lemma}[Comparison-set invariant]\label{lem:Q-invariant}
The set $Q$ can be updated as above so that~\eqref{eq:Q-rank} holds after every processed step. Moreover, every element of $Q$ is unprocessed, external, and feasible to add to $M$.
\end{lemma}

\begin{proof}
The claims hold initially because $Q=O$ is independent and $M=B=\varnothing$. Suppose first that the current element $e$ belongs to $Q$. Before its color is revealed,~\eqref{eq:Q-rank} implies $e\notin\cl(M\cup B)$. If the color is ignored, removing $e$ from $Q$ clearly preserves~\eqref{eq:Q-rank}. If $e$ is added to $M$ or to $B$, then $\rank(M\cup B)$ increases by one while $|Q|$ decreases by one, so the same equality is preserved.

Now suppose $e\notin Q$. If the step is internal, adding $e$ to $M$ or to $B$ does not change the closure of $M\cup B$, so~\eqref{eq:Q-rank} remains true without changing $Q$. If the step is external and $e$ is added to $M$ or to $B$, let $J$ be a basis of the old $M\cup B$. Since a basis has the same closure as the set it spans,~\eqref{eq:Q-rank} implies that $J\cup Q$ is independent. Also $J+e$ is independent because the step is external. If $J\cup Q+e$ is independent, then~\eqref{eq:Q-rank} remains true without changing $Q$. Otherwise, $J\cup Q+e$ has rank $|J|+|Q|$, because $J\cup Q$ is independent. Extend $J+e$ to a basis of $J\cup Q+e$. This basis contains $J+e$ and exactly $|Q|-1$ elements of $Q$, so it is $J+e\cup(Q-o)$ for some $o\in Q$. Replacing $Q$ by $Q-o$ therefore restores~\eqref{eq:Q-rank}.

Finally, $Q$ contains only elements of the original optimal set that have not yet been processed, because a processed member is removed immediately. Equation~\eqref{eq:Q-rank} implies that each $q\in Q$ is outside $\cl(M\cup B)$, hence outside $\cl(M)$, so it is external and feasible to add to $M$.
\end{proof}

\revision{To compare the optimum with the learning and output sets, we use the residual-value potential}
\begin{equation}\label{eq:Psi}
\Psi=f(M\cup A\cup Q)-f(M\cup A).
\end{equation}
At the beginning, $M=A=\varnothing$ and $Q=O$, so $\Psi=f(O)$. At the end,~\eqref{eq:Q-rank} guarantees that every remaining element of $Q$ is feasible to add to the final $M$. The process has stopped because no feasible element has positive marginal value with respect to $M$. Hence every element of $Q$ has zero marginal value with respect to $M$, and by submodularity the whole set $Q$ has zero marginal value with respect to $M\cup A$. Therefore the final value of $\Psi$ is zero.

Consider one processed element $e$ before its color and retention bit are revealed. Besides its weight
\(
w_e=f(e\mid M),
\)
define
\begin{equation}\label{eq:u-I}
u_e=f(e\mid M\cup A),
\qquad
I_e=\mathbf 1\{B+e\in\mathcal I\}.
\end{equation}
Monotonicity and submodularity give
\begin{equation}\label{eq:u-order}
0\le u_e\le w_e,
\qquad
u_e\le f(e\mid A).
\end{equation}

We next bound the expected decrease of $\Psi$ at a single step.

\begin{lemma}[Residual-optimum decrease]\label{lem:residual-decrease}
Condition on the state immediately before the color and retention bit of the current element $e$ are revealed. Thus $M,B,A,Q,w_e,u_e,I_e$, and the internal/external classification are fixed. The expectations below are only over the unrevealed color and retention bit of $e$.
\begin{enumerate}[label=(\roman*)]
\item If $e\in Q$, then the step is external and the conditional expected decrease of $\Psi$ is at most $u_e$.
\item If $e\notin Q$ and the step is external, then the conditional expected decrease of $\Psi$ is at most
\[
(p+\beta)w_e+(p+\beta\gamma)u_e.
\]
\item If the step is internal, then the conditional expected decrease of $\Psi$ is at most
\[
(p+\beta\gamma I_e)u_e.
\]
\end{enumerate}
\end{lemma}

\begin{proof}
If $e\in Q$, the invariant~\eqref{eq:Q-rank} makes the step external. The update removes $e$ from $Q$. If $e$ enters $M$ or $A$, then the union $M\cup A\cup Q$ in the first term of~\eqref{eq:Psi} remains unchanged, whereas the second term changes from $f(M\cup A)$ to $f(M\cup A+e)$. Hence $\Psi$ decreases by $f(e\mid M\cup A)=u_e$. If $e$ enters neither set, only $Q$ loses $e$, and the decrease is
\[
f(e\mid M\cup A\cup(Q-e))\le u_e.
\]
This proves (i).

Now suppose $e\notin Q$ and the step is external. If $e$ is learning-colored, it enters $M$ and one element $o\in Q$ may be removed to restore~\eqref{eq:Q-rank}. If $e$ is reference-colored and its retention bit is one, it  enters $A$, and again one element $o\in Q$ may be removed. In either event that adds $e$ to $M\cup A$, write $R=M\cup A$ for the old set. The decrease of $\Psi$ is at most
\begin{align*}
&f(R\cup Q)-f(R)-f(R+e\cup(Q-o))+f(R+e)\\
&\qquad\le u_e+f(o\mid R\cup(Q-o))\\
&\qquad\le u_e+f(o\mid M)\\
&\qquad\le u_e+w_e.
\end{align*}
The last inequality holds because $o$ is unprocessed, is feasible to add to $M$ by~\eqref{eq:Q-rank}, and $e$ was chosen with maximum current marginal value. If a reference-colored $e$ is not retained in $A$, then only $Q$ may lose $o$, and the decrease is at most $w_e$. Thus the conditional expected decrease is at most
\[
p(u_e+w_e)+\beta\gamma(u_e+w_e)+\beta(1-\gamma)w_e
=(p+\beta)w_e+(p+\beta\gamma)u_e.
\]
This proves (ii).

Finally suppose the step is internal. Then $e\notin Q$, and no element is removed from $Q$ because adding $e$ does not change the closure of $M\cup B$. If $e$ is learning-colored, $M\cup A$ gains $e$ and the decrease of $\Psi$ is at most $u_e$. If $e$ is reference-colored, $A$ gains $e$ only when $I_e=1$ and the retention bit is one, again causing decrease at most $u_e$. This proves (iii).
\end{proof}

\begin{lemma}[Upper bound]\label{lem:joint-residual}
For the three-color process,
\begin{equation}\label{eq:joint-residual}
f(O)
\le
\bigl(\max\{1,\,2p+\beta+\beta\gamma\}+\beta\bigr)
\mathbb E\!\left[\sum I_eu_e\right]
+(p+2\beta)\,\mathbb E\!\left[\sum (w_e-u_e)\right].
\end{equation}
\end{lemma}

\begin{proof}
For an external step not in $Q$, the bound in Lemma~\ref{lem:residual-decrease}(ii) can be rewritten as
\[
(2p+\beta+\beta\gamma)u_e+(p+\beta)(w_e-u_e).
\]
For an external step in $Q$, Lemma~\ref{lem:residual-decrease}(i) is at most $u_e$. Thus, for any external step, the decrease is at most
\[
\max\{1,\,2p+\beta+\beta\gamma\}u_e+(p+\beta)(w_e-u_e).
\]
Summing all decreases of $\Psi$, using its initial value $f(O)$ and final value zero, gives
\begin{align}
 f(O)\le{}&
 \max\{1,\,2p+\beta+\beta\gamma\}\,
 \mathbb E\!\left[\sum_{\mathrm{external}}u_e\right]
 +(p+\beta)\,\mathbb E\!\left[\sum_{\mathrm{external}}(w_e-u_e)\right]\notag\\
 &+p\,\mathbb E\!\left[\sum_{\mathrm{internal}}u_e\right]
 +\beta\gamma\,\mathbb E\!\left[\sum_{\mathrm{internal}}I_eu_e\right].
\label{eq:Psi-sum}
\end{align}

The rank charge in Lemma \ref{lem:rank-charge} converts the internal $u_e$-term into external terms. Since $w_e=u_e+(w_e-u_e)$, Lemma~\ref{lem:rank-charge} implies
\begin{align*}
p\,\mathbb E\!\left[\sum_{\mathrm{internal}}u_e\right]
\le{}&
\beta\,\mathbb E\!\left[\sum_{\mathrm{external}}u_e\right]
+\beta\,\mathbb E\!\left[\sum_{\mathrm{external}}(w_e-u_e)\right]-p\,\mathbb E\!\left[\sum_{\mathrm{internal}}(w_e-u_e)\right].
\end{align*}
Substituting this into~\eqref{eq:Psi-sum}, dropping the final nonpositive term, using $I_e=1$ at every external step, and using
\[
\max\{1,\,2p+\beta+\beta\gamma\}+\beta\ge\beta\gamma,
\]
gives~\eqref{eq:joint-residual}.
\end{proof}

\subsection{Competitive ratio}\label{subsec:weight-redundancy}

\revision{It remains to control the second term in~\eqref{eq:joint-residual}. Ma, Tang, and Wang~\cite[Lemma~11]{ma2013} use a potential containing the difference between the value of the learning set and that of its union with the candidate set. Here we apply this idea to the retained output $A$ through the nonnegative potential}
\begin{equation}\label{eq:Phi}
\Phi=f(M\cup A)-f(M).
\end{equation}
At one step, if $e$ is learning-colored, it is added to $M$ and changes the potential by $u_e-w_e$. If $e$ is reference-colored, it is added to $A$ only when $I_e=1$ (meaning that $B+e\in\mathcal I$) and its retention bit is one. In that event, the potential changes by $u_e$. In every other case, the potential is unchanged. 
Therefore, the conditional expected increment of $\Phi$ at one step  is \[
\mathbb E[\Delta\Phi\mid\text{current state}]
=p(u_e-w_e)+\beta\gamma I_eu_e,
\]
where the current state means the complete state immediately before the color and retention bit of $e$ are revealed. 
The potential starts at zero and is nonnegative at the end. Hence
\begin{equation}\label{eq:gap-payment}
p\,\mathbb E\!\left[\sum(w_e-u_e)\right]
\le
\beta\gamma\,\mathbb E\!\left[\sum I_eu_e\right].
\end{equation}

The actual output value is accumulated through the marginal value $f(e\mid A)$. A step adds $e$ to $A$ with conditional probability $\beta\gamma I_e$, so
\begin{equation}\label{eq:A-value}
\mathbb E[f(A)]
=\beta\gamma\,\mathbb E\!\left[\sum I_e f(e\mid A)\right]
\ge
\beta\gamma\,\mathbb E\!\left[\sum I_eu_e\right],
\end{equation}
where the inequality uses $u_e\le f(e\mid A)$ from~\eqref{eq:u-order}.

Combining Lemma~\ref{lem:joint-residual} with~\eqref{eq:gap-payment} and~\eqref{eq:A-value} gives the general competitive guarantee.

\begin{theorem}[Competitive ratio]\label{thm:general}\label{thm:main}
For every $p,s\in(0,1)$, let
\[
\beta=(1-p)s,
\qquad
\gamma=1-\exp\left(-\frac{1-s}{s}\right).
\]
Then Algorithm~\ref{alg:main} is $C(p,s)$-competitive, where
\begin{equation}\label{eq:general-C}
C(p,s)=
\frac{\max\{1,\,2p+\beta+\beta\gamma\}+\beta}{\beta\gamma}
+1+\frac{2\beta}{p}.
\end{equation}
In particular, for
\[
p=\frac5{18},
\qquad
s=\frac{23}{70},
\]
the algorithm is $C$-competitive, where
\begin{equation}\label{eq:main-constant}
C=
\frac{1559}{299(1-e^{-47/23})}
+\frac{474}{175}
=8.6988177308\ldots.
\end{equation}
\end{theorem}

\begin{proof}
By Lemma~\ref{lem:joint-residual} and~\eqref{eq:gap-payment},
\begin{align*}
f(O)
&\le
\left(\max\{1,\,2p+\beta+\beta\gamma\}+\beta
+(p+2\beta)\frac{\beta\gamma}{p}\right)
\mathbb E\!\left[\sum I_eu_e\right]\\
&=
\left(
\frac{\max\{1,\,2p+\beta+\beta\gamma\}+\beta}{\beta\gamma}
+1+\frac{2\beta}{p}
\right)
\beta\gamma\,\mathbb E\!\left[\sum I_eu_e\right],
\end{align*}
which is no more than $C(p,s)\mathbb E[f(A)]$ by
Equation~\eqref{eq:A-value}. Thus the general guarantee follows.

For the stated parameters $p$ and $s$,
\[
\beta=(1-p)s=\frac{299}{1260},
\qquad
\gamma=1-e^{-47/23}.
\]
Moreover, it is easy to verify that
\[
2p+\beta+\beta\gamma
=
\frac{1298-299e^{-47/23}}{1260}<1.
\]
Hence
\[
\max\{1,\,2p+\beta+\beta\gamma\}=1.
\]
Substituting these values into~\eqref{eq:general-C} gives
\begin{align*}
C
&=
\frac{1+\beta}{\beta\gamma}+1+\frac{2\beta}{p}=
\frac{1559}{299(1-e^{-47/23})}+\frac{474}{175},
\end{align*}
which is~\eqref{eq:main-constant}.
\end{proof}

We count value queries and independence queries in Algorithm~\ref{alg:main} separately. It does not need an exact offline submodular optimization oracle.
Recall that $r$ is the matroid rank. 

\begin{theorem}[Query complexity]\label{thm:queries}
     Algorithm~\ref{alg:main} uses $O(nr)$ value queries and $O(n^2)$ independence queries. All queries involve only elements that have already arrived.
\end{theorem}

\begin{proof}
    We first show that 
after the learning phase in Algorithm~\ref{alg:main}, the frozen weight of one arriving element can be computed using $O(r)$ value queries and $O(r)$ independence queries. 
Store the ordered greedy sequence $m_1,\ldots,m_k$, its prefixes $M_i$, and the marginal value at which every $m_{i+1}$ was selected from $M_i$. Scan the stored prefixes in order. At prefix $M_i$, one independence query checks whether $M_i+e$ is independent and one value query evaluates $f(e\mid M_i)$. If $e$ is feasible and its marginal is larger than the stored marginal of $m_{i+1}$, or is tied and precedes $m_{i+1}$ in the fixed label order, then the hypothetical run selects $e$ next and its weight is known. At the last prefix, a feasible element of positive marginal is selected. If $e$ ever becomes infeasible, it remains infeasible at every later prefix. Thus at most $|M|+1\le r+1$ prefixes are inspected. Since each inspected prefix uses at most one value query and one independence query, this proves the $O(r)$ query bound for one frozen weight.

Running submodular greedy on the learning sample can be implemented by at most $r$ greedy rounds, each scanning at most $n$ available elements. This uses $O(nr)$ value queries and $O(nr)$ independence queries.
Computing all post-learning frozen weights therefore uses another $O(nr)$ queries of each type. The inner Dynamic Thinning instance uses no further value queries, and by Lemma~\ref{lem:dt-query}, it uses $O(n^2)$ independence queries in the worst case. Since $r\le n$, the total number of independence queries is $O(n^2)$, while the total number of value queries is $O(nr)$.
\end{proof}

\section{Conclusion}\label{sec:conclusion}

Replacing the single-level Dynamic Thinning subroutine by Multilevel Dynamic Thinning~\cite{joyce2026} in our framework gives a slightly improved competitive factor of about $8.6$. The improvement is modest, while both the analysis and the implementation become substantially more involved. It remains an interesting question whether a different approach can improve the competitive ratio further without introducing comparable complexity.

Another direction is to exploit the structure of special matroid classes. In particular, laminar and graphic matroids may admit stronger guarantees than arbitrary matroids. For these two classes, we obtain several improved bounds in the appendix. Determining the best competitive ratios for such structured matroids remains open.

\paragraph{AI statement.}
 The author used GPT-5.6 Sol and GPT-6 Astra to help design the algorithm and
develop its proofs. The author has checked the proofs and takes full responsibility for their correctness
and for the contents of this paper.


\clearpage
\appendix
\begin{center}
\section*{Appendix}
\end{center}
\section{Guarantees for Laminar and Graphic Matroids}\label{app:structured}
\setcounter{equation}{0}
\renewcommand{\theequation}{\thesection.\arabic{equation}}
\renewcommand{\theHequation}{\thesection.\arabic{equation}}

For the monotone submodular matroid secretary problem, we give polynomial-time randomized online algorithms with competitive ratios
\[
 C_{\mathrm{lam}}=\frac{31250}{4743}\approx6.589
 \qquad\text{and}\qquad
 C_{\mathrm{gr}}=\frac{546750}{80771}\approx6.769
\]
for laminar and graphic matroids, respectively. The laminar result holds without any restriction on the depth or branching of the laminar family. The graphic result allows undirected multigraphs, including parallel edges, and assumes that the endpoints of an edge are revealed when it arrives. Both algorithms use $O(n^2r)$ value queries and independence queries. 

\revision{Unlike Algorithm~\ref{alg:main}, the algorithms in this appendix recompute randomized greedy sets from arrival prefixes. They combine a known sampling guarantee with existing matroid-specific acceptance arguments.}

\begin{definition}[Laminar matroid]\label{def:laminar-appendix}
A family $\mathcal L\subseteq2^E$ is \emph{laminar} if any two sets in the family are disjoint or one contains the other. Given a nonnegative integer capacity $b_F$ for each $F\in\mathcal L$, the associated laminar matroid has independent sets
\[
 \mathcal I=\{S\subseteq E:|S\cap F|\le b_F\text{ for every }F\in\mathcal L\}.
\]
\end{definition}

\begin{definition}[Graphic matroid]\label{def:graphic-appendix}
Let $G=(V,E)$ be an undirected multigraph. Its \emph{graphic matroid} has ground set $E$, and an edge set is independent exactly when it is a forest. A self-loop cannot belong to any forest and hence to any independent set. Two parallel edges form a cycle of length two.
\end{definition}

Throughout this appendix, fix an offline optimal independent set
\[
 O\in\arg\max\{f(S):S\in\mathcal I\}.
\]
All label orders are fixed independently of the random arrivals. The algorithms use the shared submodular greedy subroutine described below.

\subsection{A shared submodular greedy subroutine}\label{app:prefix-subroutine}


To obtain a uniform guarantee for every finite ground set, add zero-value dummy elements that cannot belong to any independent set, so that the extended ground set has size
\[
 N=18\left\lceil\frac{\max\{n,31250000\}}{18}\right\rceil.
\]
For a set $S$ of real and dummy elements, extend the objective by $f^+(S)=f(S\cap E)$. We write $f$ for this extension below. It remains normalized, monotone, and submodular, and the optimum value is still $f(O)$.

Choose the $n$ positions of the real elements uniformly among the $N$ slots, independently of the real arrival order, and uniformly permute the labeled dummy elements over the other positions. Interleaving the real arrivals with these dummy elements produces a uniformly random permutation of the extended ground set. Dummy elements are always rejected and require no oracle calls. Let $e_i$ be the element in slot $i$ and $H_i=\{e_1,\ldots,e_i\}$.

\Needspace{22\baselineskip}
We first bound the value of submodular greedy on an independent random sample.
This lemma is a consequence of the monotone SampleGreedy guarantee of Feldman, Harshaw, and Karbasi~\cite{feldmanharshawkarbasi2017}. 
 We retain a self-contained proof below.

\begin{lemma}[Submodular greedy on a random sample \cite{feldmanharshawkarbasi2017}]\label{lem:sample-greedy-appendix}
Let $R$ contain each element independently with probability $q\in[0,1/2]$, and let $\operatorname{Greedy}(R)$ be the submodular greedy set defined in Section~\ref{sec:prelim}. Then
\[
 \mathbb E[f(\operatorname{Greedy}(R))]\ge qf(O).
\]
\end{lemma}

\begin{proof}
The claim is immediate when $q=0$, so assume $q>0$. Reveal the membership bits of $R$ only as needed. Maintain $M$ and $Q\subseteq O\setminus M$ such that $M\cup Q$ is independent and every element of $Q$ is unprocessed. Initially $M=\varnothing$ and $Q=O$. At each step, choose an unprocessed element $e$ that is feasible for $M$ and has maximum positive marginal value
\[
 w_e=f(e\mid M),
\]
breaking ties by the fixed label order. Then reveal whether $e\in R$.

If $e\in Q$, remove it from $Q$ and add it to $M$ if it belongs to $R$. If $e\notin Q$ and $e\in R$, add $e$ to $M$. When necessary, remove one element $o\in Q$ to restore independence of $M\cup Q$. To see that one removal suffices, start with the independent sets $M+e$ and the old $M\cup Q$ and repeatedly apply the matroid augmentation property. This gives an independent set containing $M+e$ and all but at most one element of the old $Q$. Every possible removed element $o$ is unprocessed and feasible for the old $M$, so
\[
 f(o\mid M)\le w_e.
\]
If $e\notin Q$ and $e\notin R$, neither set changes. In every case, mark $e$ as processed.

The resulting $M$ has the same distribution as $\operatorname{Greedy}(R)$. Indeed, between two insertions into $M$, every processed element is nonsampled, and the next sampled element encountered is the maximum-marginal feasible sampled element. An element that becomes infeasible remains infeasible as $M$ grows.

Consider
\[
 \Psi=f(M\cup Q)-f(M)\ge0.
\]
For a step processing $e$, let $D_e$ be the value of $\Psi$ before the step minus its value afterward. Condition on the complete history immediately before revealing the membership bit of $e$. Under this conditioning, $M,Q,e,w_e$ are fixed and $\Pr[e\in R]=q$.

If $e\in Q$ and is sampled, the set $M\cup Q$ does not change and $D_e=w_e$. If $e\in Q$ and is not sampled, then
\[
 D_e=f(e\mid M\cup(Q-e))\le w_e.
\]
If $e\notin Q$ and is sampled, the increase in $f(M)$ is $w_e$. Any loss from removing $o$ from $Q$ is at most $f(o\mid M)\le w_e$, while adding $e$ to the union cannot decrease its value. Thus $D_e\le2w_e$. The remaining case has $D_e=0$. Consequently,
\[
 \mathbb E[D_e\mid\text{history}]
 \le w_e
 =\frac1q\mathbb E[\text{increase in }f(M)\mid\text{history}],
\]
where the bound for $e\notin Q$ uses $2q\le1$.

At termination, every element remaining in $Q$ is feasible for $M$ and has zero marginal value with respect to $M$. Submodularity then gives $f(M\cup Q)=f(M)$, so the final potential is zero. Summing the conditional inequalities over the at most $N$ processed steps and taking expectations yields
\[
 f(O)
 =\mathbb E\!\left[\sum_eD_e\right]
 \le\frac1q\mathbb E[f(M)].
\]
This proves the claim.
\end{proof}

The following construction realizes almost the same value guarantee using only a random prefix.

\begin{samepage}
\begin{lemma}[An independent set from a random prefix]\label{lem:prefix-value-appendix}
For a slot $i>4N/9$, set
\[
 q_i=\min\left\{\frac iN,\frac12\right\}-\frac1{500}.
\]
Independently of the arrivals and all other slots, draw $K_i\sim\operatorname{Bin}(N,q_i)$. If $K_i>i$, set $B_i=\varnothing$. Otherwise choose a uniformly random $K_i$-element subset $R_i\subseteq H_i$ and let $B_i=\operatorname{Greedy}(R_i)$. Then $B_i$ is independent and
\begin{equation}\label{eq:app-prefix-value}
 \mathbb E[f(B_i)]
 \ge\left(\min\left\{\frac iN,\frac12\right\}-\frac1{250}\right)f(O).
\end{equation}
\end{lemma}
\end{samepage}

\begin{proof}
Condition on $K_i=k\le i$. Since $H_i$ is a uniformly random $i$-subset of the extended ground set and $R_i$ is a uniformly random $k$-subset of $H_i$, the set $R_i$ is a uniformly random $k$-subset of the whole ground set. Its distribution therefore agrees with that of a global independent Bernoulli-$q_i$ sample $\widetilde R_i$, conditional on $|\widetilde R_i|=k$. Lemma~\ref{lem:sample-greedy-appendix} gives
\[
 \mathbb E[f(\operatorname{Greedy}(\widetilde R_i))]\ge q_if(O).
\]
Each greedy set is independent and has value at most $f(O)$. Setting $B_i=\varnothing$ when $K_i>i$ therefore loses at most $\Pr[K_i>i]f(O)$. Moreover,
\[
 i-\mathbb E[K_i]\ge\frac N{500},
 \qquad
 \operatorname{Var}(K_i)\le\frac N4.
\]
By Chebyshev's inequality,
\[
 \Pr[K_i>i]
 \le\frac{N/4}{(N/500)^2}
 =\frac{62500}{N}
 \le\frac1{500}.
\]
It follows that
\[
 \mathbb E[f(B_i)]
 \ge\left(q_i-\frac1{500}\right)f(O),
\]
which is~\eqref{eq:app-prefix-value}.
\end{proof}

For the conditioning arguments below, represent the randomness at slot $i$ by an independent seed $\eta_i$. It specifies $K_i$ and a uniformly random permutation of the indices $1,\ldots,i$. List $H_i$ in the fixed label order and use the first $K_i$ indices of this permutation to choose $R_i$. Thus $B_i$ is a deterministic function of the set $H_i$ and $\eta_i$, independently of the order within $H_i$.


\subsection{Laminar matroids}\label{app:laminar}

Set $\ell=N/2$ and reject the first $\ell$ slots. At each later real arrival $e_i$, compute $B_i$ using Lemma~\ref{lem:prefix-value-appendix}. Call $e_i$ a \emph{candidate} if $e_i\in B_i$. Accept a candidate exactly when adding it to the current accepted set preserves independence, and reject every other arrival. 

The laminar sets containing a fixed element form a chain. We use the following counting lemma to bound the probability that any constraint in this chain prevents acceptance.
\revision{The lemma gives a finite-slot form of the counting identity in B\'erczi et al.~\cite{berczilabeling2025}. We give an alternative proof using cyclic parking.}

\begin{samepage}
\begin{lemma}[A counting bound for nested capacities]\label{lem:parking-appendix}
Let $\ell\ge1$ and $0\le m\le\ell$ be integers. Let
\[
 U_j\sim\operatorname{Unif}\{1,\ldots,j\},
 \qquad j=\ell+1,\ldots,\ell+m,
\]
be independent. Then
\begin{equation}\label{eq:app-parking}
 \Pr\!\left[\exists b\in\{1,\ldots,m\}:
   |\{j:U_j\le b\}|\ge b\right]
 =\sum_{j=\ell+1}^{\ell+m}\frac1j.
\end{equation}
\end{lemma}
\end{samepage}

\begin{proof}
For $m=0$, both sides are zero. Assume $m>0$. Call an index $j$ active if $U_j\le m$, and let $Z$ be the number of active indices. Conditional on the set of active indices, their values are independent and uniform on $\{1,\ldots,m\}$. Values at inactive indices cannot affect the event in~\eqref{eq:app-parking}.

Given $Z=a$ active values, regard them as the preferred positions of $a$ cars on $m$ cyclically ordered parking spaces. Process the cars in a fixed order, moving each car in the decreasing cyclic direction until it finds an empty space. Exactly $a$ spaces become occupied, so rotation symmetry gives probability $a/m$ that space $1$ is occupied.

Space $1$ remains empty exactly when
\[
 |\{\text{preferences at most }b\}|<b
 \qquad(1\le b\le m).
\]
Necessity follows because every car with preference at most $b$ must then park in one of spaces $2,\ldots,b$. For sufficiency, sort the preferences as $p_1\le\cdots\le p_a$. The displayed inequalities imply $p_k\ge k+1$, so assigning car $k$ to space $k+1$ gives a feasible assignment in spaces $2,\ldots,m$. The decreasing greedy parking rule also succeeds whenever such an assignment exists: assigning a car to the largest available space no greater than its preference preserves a feasible assignment for the remaining cars, by swapping with the car assigned to that space when necessary. Hence this rule never reaches space $1$.

Thus the event in~\eqref{eq:app-parking} is exactly that space $1$ is occupied. Its probability, conditional on the active index set, is $Z/m$. Taking expectations gives
\[
 \Pr[\text{event}]
 =\frac{\mathbb E[Z]}m
 =\frac1m\sum_{j=\ell+1}^{\ell+m}\frac mj
 =\sum_{j=\ell+1}^{\ell+m}\frac1j.
\]
\end{proof}

\begin{lemma}[Conditional feasibility of a laminar candidate]\label{lem:laminar-feasible-appendix}
Fix $i>\ell$ and condition on $H_i$, the identity of $e_i$, all future elements $e_{i+1},\ldots,e_N$, and all seeds $\eta_j$. If $e_i\in B_i$, then
\begin{equation}\label{eq:app-lam-feasibility}
 \Pr[e_i\text{ is feasible for the current accepted set}
       \mid H_i,e_i,e_{i+1},\ldots,e_N,(\eta_j)_j]
 \ge1-\sum_{j=\ell+1}^{i-1}\frac1j.
\end{equation}
The remaining randomness is the uniformly random order of $H_i-e_i$ in the first $i-1$ slots.
\end{lemma}

\begin{proof}
List the laminar sets containing $e_i$ as a chain
\[
 F_1\subsetneq\cdots\subsetneq F_t,
\]
and write $b_k=b_{F_k}$. Each $b_k$ is positive because $e_i\in B_i$ and $B_i$ is independent. If the chain is empty, the element is always feasible and the claim is immediate. The capacities need not be strictly increasing.

Reveal the earlier order backwards. At reverse step $j$, the set $H_j$ and seed $\eta_j$ are determined, while $e_j$ is uniformly distributed over $H_j$. In particular, $B_j$ is already determined. For each $x\in H_j$, assign the integer
\[
 \min\{b_k:x\in B_j\cap F_k\},
\]
or $\infty$ if this set is empty. For every positive integer $b$, at most $b$ elements receive an integer at most $b$. Indeed, all these elements lie in $B_j$ and in the largest chain set whose capacity is at most $b$, and $B_j$ satisfies that constraint.

Sort $H_j$ by these integers, breaking ties by the fixed label order, and let $U_j$ be the position of $e_j$ in this sorted list. Conditional on everything revealed so far in the reverse order, $U_j$ is uniform on $\{1,\ldots,j\}$. Its conditional law does not depend on the revealed later positions, so the variables $U_j$ are mutually independent. If $e_j$ is a candidate in a chain set of capacity $b$, then $U_j\le b$.

If $e_i$ is infeasible, some chain set of capacity $b$ already contains $b$ accepted elements. Every such element is an earlier candidate, so at least $b$ indices $j\in\{\ell+1,\ldots,i-1\}$ satisfy $U_j\le b$. There are only $m=i-\ell-1\le\ell$ earlier slots after the initial rejection phase, so only capacities $b\le m$ can block $e_i$. Lemma~\ref{lem:parking-appendix} bounds the probability of this event by
\[
 \sum_{j=\ell+1}^{i-1}\frac1j,
\]
proving~\eqref{eq:app-lam-feasibility}.
\end{proof}

\begin{theorem}[Laminar guarantee]\label{thm:laminar-appendix}
For every laminar matroid and every nonnegative monotone submodular objective, the online algorithm above has competitive ratio at most
\[
 C_{\mathrm{lam}}=\frac{31250}{4743}\approx6.589.
\]
It uses $O(n^2r)$ value queries and independence queries.
\end{theorem}

\begin{proof}
Let $A$ be the final accepted set. For each $i>\ell$, define the candidate set at slot $i$ and the set of future candidates by
\[
 X_i=
 \begin{cases}
  \{e_i\},&e_i\in B_i,\\
  \varnothing,&e_i\notin B_i,
 \end{cases}
 \qquad
 T_{>i}=\bigcup_{j>i}X_j.
\]
Let
\[
 d_i=f(T_{>i}\cup X_i)-f(T_{>i}),
 \qquad x_i=\mathbb E[d_i].
\]
These quantities are used only in the analysis, and $d_i\ge0$ by monotonicity. For every realization, add the accepted elements in decreasing slot order. The previously added accepted elements form a subset of $T_{>i}$, so submodularity gives
\begin{equation}\label{eq:app-lam-output}
 f(A)\ge\sum_{i>\ell}\mathbf1\{e_i\in A\}d_i.
\end{equation}

Condition on $H_i$, all future elements, and all seeds, but not on $e_i$. This fixes $B_i$ and $T_{>i}$: every future prefix set is determined by $H_i$ and the future elements, regardless of the order within $H_i$. The current element $e_i$ is still uniform in $H_i$. Therefore
\[
 \mathbb E[d_i\mid H_i,e_{i+1},\ldots,e_N,(\eta_j)_j]
 =\frac1i\sum_{e\in B_i}f(e\mid T_{>i})
 \ge\frac{f(B_i)-f(T_{>i})}{i}.
\]
The inequality follows from submodularity and monotonicity. Since
\[
 f(T_{>i})=\sum_{j>i}d_j,
\]
taking expectations gives
\begin{equation}\label{eq:app-lam-recurrence}
 x_i+\frac1i\sum_{j>i}x_j
 \ge\frac{\mathbb E[f(B_i)]}{i}.
\end{equation}
Because $i>N/2$, Lemma~\ref{lem:prefix-value-appendix} gives
\begin{equation}\label{eq:app-lam-prefix}
 \mathbb E[f(B_i)]\ge\frac{62}{125}f(O).
\end{equation}

Now also condition on $e_i$. If it is a candidate, $d_i$ is fixed under this conditioning, and Lemma~\ref{lem:laminar-feasible-appendix} gives acceptance probability at least
\[
 q_i^{\mathrm{lam}}=1-\sum_{j=\ell+1}^{i-1}\frac1j.
\]
If it is not a candidate, $d_i=0$. Since $d_i\ge0$, multiplying by the conditional acceptance probability, taking expectations, and summing~\eqref{eq:app-lam-output} yields
\begin{equation}\label{eq:app-lam-acceptance}
 \mathbb E[f(A)]\ge\sum_{i=\ell+1}^{N}q_i^{\mathrm{lam}}x_i.
\end{equation}

Set
\[
 a_i=\frac{2\ell}{i-1}-1
 \qquad(i=\ell+1,\ldots,N).
\]
These coefficients are nonnegative because $N=2\ell$. Using
\[
 \frac{a_j}{j}=\frac{2\ell}{j(j-1)}-\frac1j,
\]
we obtain
\begin{equation}\label{eq:app-lam-coefficients}
 a_i+\sum_{j=\ell+1}^{i-1}\frac{a_j}{j}
 =q_i^{\mathrm{lam}},
 \qquad
 \sum_{i=\ell+1}^{N}\frac{a_i}{i}
 =1-\sum_{i=\ell+1}^{N}\frac1i.
\end{equation}
Multiplying~\eqref{eq:app-lam-recurrence} by $a_i$ and summing shows that its left-hand side is exactly $\sum_iq_i^{\mathrm{lam}}x_i$. Thus~\eqref{eq:app-lam-prefix}--\eqref{eq:app-lam-coefficients} give
\begin{equation}\label{eq:app-lam-value}
 \mathbb E[f(A)]
 \ge\frac{62}{125}\left(1-\sum_{i=N/2+1}^{N}\frac1i\right)f(O)
 \ge\frac{62}{125}(1-\log2)f(O).
\end{equation}
The last inequality uses
\[
 \sum_{i=N/2+1}^{N}\frac1i
 \le\int_{N/2}^{N}\frac{dt}{t}=\log2.
\]

Finally,
\[
 \log2
 =2\sum_{k\ge0}\frac{1}{(2k+1)3^{2k+1}}
 <\frac{56}{81}+\frac1{540}
 =\frac{1123}{1620}
 <\frac{347}{500}.
\]
Here the first two terms sum to $56/81$, and the remaining terms are bounded by the geometric series
\[
 \frac25\sum_{k\ge2}\frac1{3^{2k+1}}=\frac1{540}.
\]
Consequently, including the case $f(O)=0$,
\[
 \mathbb E[f(A)]
 \ge\frac{62}{125}\cdot\frac{153}{500}f(O)
 =\frac{4743}{31250}f(O).
\]
This is the claimed competitive guarantee. The query bound follows from the shared subroutine and at most one additional independence query per real arrival.
\end{proof}

\subsection{Graphic matroids}\label{app:graphic}

Assume that the endpoints of each real edge are revealed when it arrives, and fix a label order on vertices. The dummy elements can be represented by self-loops at an auxiliary vertex. Set $\ell=4N/9$ and reject the first $\ell$ slots.

At each later real arrival $e_i$, compute the forest $B_i$ using Lemma~\ref{lem:prefix-value-appendix}. Root each nontrivial tree of $B_i$ at its smallest-labeled vertex and orient its edges away from the root. If $e_i\in B_i$, call it a candidate and use its orientation in $B_i$. The algorithm maintains a set of marked vertices, initially empty. Accept a candidate $e_i:u\to v$ exactly when neither $u$ nor $v$ is marked; upon acceptance, mark $v$. Reject all other arrivals.

\revision{This canonical orientation and acceptance rule are those of the graphic-matroid algorithm of Soto, Turkieltaub, and Verdugo~\cite{soto2021}. We replace their prefix optimum forest by the randomized submodular greedy forest $B_i$. }
\begin{lemma}[Feasibility of the graphic algorithm]\label{lem:graphic-feasible-appendix}
The accepted set is always a forest, including when the graph has parallel edges.
\end{lemma}

\begin{proof}
Maintain the invariant that every connected component of the accepted graph is a tree with exactly one unmarked vertex, including isolated vertices. It holds initially. An accepted edge has two unmarked endpoints, which cannot belong to the same component by the invariant. It therefore joins two different trees. Its head becomes marked and its tail remains unmarked, so the merged tree again has exactly one unmarked vertex. A graph loop is never a candidate because $B_i$ is a forest, and the second of two parallel edges cannot be accepted because one endpoint is already marked.
\end{proof}

\begin{theorem}[Graphic guarantee]\label{thm:graphic-appendix}
For every graphic matroid represented by an undirected multigraph whose edge endpoints are revealed on arrival, and every nonnegative monotone submodular objective, the online algorithm above has competitive ratio at most
\[
 C_{\mathrm{gr}}=\frac{546750}{80771}\approx6.769.
\]
It uses $O(n^2r)$ value queries and independence queries.
\end{theorem}

\begin{proof}
Let $A$ be the output. For analysis, consider a second process with the same candidates and orientations. It accepts a candidate only when both endpoints are unmarked, but marks the head of every candidate, whether or not it is accepted. Let $A^\circ$ be its output.  At every slot, this process has marked all previous candidate heads and hence contains every mark made by the online algorithm. Any candidate it accepts therefore also passes the online algorithm's test. By monotonicity,
\begin{equation}\label{eq:app-gr-coupling}
 A^\circ\subseteq A,
 \qquad f(A)\ge f(A^\circ).
\end{equation}

For $i>\ell$, let $S_i$ be the output of the second process started at slot $i$ with no marked vertices, using the same candidate status and orientation at every subsequent slot. In particular,
\[
 S_{N+1}=\varnothing,
 \qquad S_{\ell+1}=A^\circ.
\]
Condition on $H_i$, all future edges, and all seeds, leaving $e_i$ uniform in $H_i$. This fixes $B_i$ and $S=S_{i+1}$ because future prefix sets do not depend on the order within $H_i$.

If $e_i\notin B_i$, then $S_i=S$. Otherwise, write its orientation as $e:u\to v$ and let $S_i(e)$ denote the resulting output. The process accepts $e$ and marks $v$. Every subsequent candidate marks its head regardless of acceptance, so starting the suffix with $v$ marked changes no other marks and removes exactly the edges of $S$ incident to $v$. Consequently,
\begin{equation}\label{eq:app-gr-suffix}
 S_i(e)=\{e\}\cup(S\setminus D_v),
 \qquad D_v=\{x\in S:x\text{ is incident to }v\}.
\end{equation}

List $S$ in the fixed edge-label order as $x_1,\ldots,x_k$, and define
\[
 c_{x_j}=f(x_j\mid\{x_1,\ldots,x_{j-1}\}).
\]
These marginal values are nonnegative and sum to $f(S)$. For any $D\subseteq S$, submodularity gives
\[
 f(S\setminus D)\ge\sum_{x\in S\setminus D}c_x
 =f(S)-\sum_{x\in D}c_x,
 \qquad
 f(e\mid S\setminus D)\ge f(e\mid S).
\]
Adding these inequalities yields
\begin{equation}\label{eq:app-gr-deletion}
 f((S\setminus D)+e)
 \ge f(S)+f(e\mid S)-\sum_{x\in D}c_x.
\end{equation}

Every vertex is the head of at most one edge in the rooted forest $B_i$. Thus each edge of $S$ belongs to at most two sets $D_v$ indexed by candidate heads, one for each endpoint. The sum of all deletion terms in~\eqref{eq:app-gr-deletion} is therefore at most $2f(S)$. Also,
\[
 \sum_{e\in B_i}f(e\mid S)
 \ge f(B_i\cup S)-f(S)
 \ge f(B_i)-f(S).
\]
Summing~\eqref{eq:app-gr-deletion} over $e\in B_i$ and using~\eqref{eq:app-gr-suffix} gives
\begin{equation}\label{eq:app-gr-total-change}
 \sum_{e\in B_i}\bigl[f(S_i(e))-f(S)\bigr]
 \ge f(B_i)-3f(S).
\end{equation}
Noncandidates contribute zero to this change. Since $e_i$ is uniform in $H_i$, we have, under the stated conditioning,
\[
 \mathbb E[f(S_i)\mid H_i,e_{i+1},\ldots,e_N,(\eta_j)_j]
 \ge\left(1-\frac3i\right)f(S)+\frac{f(B_i)}i.
\]
Writing $y_i=\mathbb E[f(S_i)]$ and taking expectations yields
\begin{equation}\label{eq:app-gr-recurrence}
 y_i\ge\left(1-\frac3i\right)y_{i+1}
       +\frac{\mathbb E[f(B_i)]}{i}.
\end{equation}
Lemma~\ref{lem:prefix-value-appendix} gives
\begin{equation}\label{eq:app-gr-prefix}
 \mathbb E[f(B_i)]
 \ge\left(\min\left\{\frac iN,\frac12\right\}-\frac1{250}\right)f(O).
\end{equation}

For $\ell+1\le i\le N+1$, put
\[
 P_i=\prod_{j=\ell+1}^{i-1}\left(1-\frac3j\right)
 =\frac{\ell(\ell-1)(\ell-2)}{(i-1)(i-2)(i-3)}.
\]
The factors are positive because $\ell>2$, and $P_{\ell+1}=1$. Moreover,
\[
 \frac{P_i}{i}=\frac{P_i-P_{i+1}}3,
 \qquad
 \sum_{i=\ell+1}^{N}\frac{P_i}{i}
 =\frac{1-P_{N+1}}3<\frac13.
\]
Unrolling~\eqref{eq:app-gr-recurrence}, using $y_{N+1}=0$, and substituting~\eqref{eq:app-gr-prefix} gives
\begin{equation}\label{eq:app-gr-unrolled}
 y_{\ell+1}
 \ge\left[
   \sum_{i=\ell+1}^{N}\frac{P_i}{i}
       \min\left\{\frac iN,\frac12\right\}
   -\frac1{750}
 \right]f(O).
\end{equation}

To evaluate the sum, write $h=N/2$ and $L=\ell(\ell-1)(\ell-2)$. Both $h$ and $\ell=4N/9$ are integers because $N$ is divisible by $18$. The identities
\[
 \frac1{(i-1)(i-2)(i-3)}
 =\frac12\left[
   \frac1{(i-2)(i-3)}-\frac1{(i-1)(i-2)}
 \right]
\]
and
\[
 \frac1{i(i-1)(i-2)(i-3)}
 =\frac13\left[
   \frac1{(i-1)(i-2)(i-3)}
   -\frac1{i(i-1)(i-2)}
 \right]
\]
show that
\[
 \sum_{i=\ell+1}^{h}\frac{P_i}{N}
 =\frac{\ell}{2N}-\frac{L}{2N(h-1)(h-2)}
\]
and
\[
 \frac12\sum_{i=h+1}^{N}\frac{P_i}{i}
 =\frac{L}{6}\left[
   \frac1{h(h-1)(h-2)}-\frac1{N(N-1)(N-2)}
 \right].
\]
Adding these expressions and using $h=N/2$ yields
\begin{equation}\label{eq:app-gr-telescoping}
 \sum_{i=\ell+1}^{N}\frac{P_i}{i}
       \min\left\{\frac iN,\frac12\right\}
 =\frac{\ell}{2N}
  -\frac{L}{6N}\left[
    \frac1{(N/2-1)(N/2-2)}
    +\frac1{(N-1)(N-2)}
  \right].
\end{equation}
For $j=1,2$,
\[
 \frac{\ell-j}{N/2-j}\le\frac89,
 \qquad
 \frac{\ell-j}{N-j}\le\frac49.
\]
Thus the subtracted term in~\eqref{eq:app-gr-telescoping} is at most
\[
 \frac{\ell}{6N}
 \left[\left(\frac89\right)^2+\left(\frac49\right)^2\right]
 =\frac56\left(\frac49\right)^3,
\]
and the sum is at least
\[
 \frac29-\frac56\left(\frac49\right)^3
 =\frac{326}{2187}.
\]
Combining this with~\eqref{eq:app-gr-coupling},~\eqref{eq:app-gr-unrolled}, and $S_{\ell+1}=A^\circ$ gives
\[
 \mathbb E[f(A)]
 \ge\left(\frac{326}{2187}-\frac1{750}\right)f(O)
 =\frac{80771}{546750}f(O).
\]
This is the claimed competitive guarantee. The query bound follows from the shared subroutine; rooting the prefix forests and maintaining vertex marks take polynomial time and require no additional oracle calls.
\end{proof}

\end{document}